\documentclass[reqno]{amsart}
\usepackage{amsaddr}
\usepackage[T1]{fontenc}
\usepackage[dvipsnames]{xcolor}
\usepackage[left=24mm,right=24mm,top=35mm]{geometry}
\usepackage{amsmath,amssymb,amsthm,mathtools}
\usepackage[shortlabels]{enumitem}
\usepackage[numbers,sort&compress]{natbib}
\usepackage{dsfont}
\usepackage{bm,graphicx,booktabs}
\usepackage{thmtools,thm-restate}
\usepackage{physics}
\usepackage[expansion=false]{microtype}
\usepackage{tikz}
\usetikzlibrary{arrows.meta}

\newcommand{\inner}[2]{\langle #1,#2\rangle}

\definecolor{sfpInk}{rgb}{.10,.19,.22}
\definecolor{sfpLow}{rgb}{.97,.96,.91}
\definecolor{sfpHigh}{rgb}{.09,.37,.42}
\pgfdeclarefunctionalshading{sfp-island}{\pgfpoint{0bp}{0bp}}{\pgfpoint{100bp}{100bp}}{}{
  25 sub 50 div exch 25 sub 50 div exch
  1 1 2.718281828 3 index 0.22 0.07 360 8 index 0.64 sub mul cos mul add 0.05 720 8 index 0.64 sub mul cos mul add sub dup 0 lt { pop 0 } if 0.075 div dup mul neg exp sub 1 0.92 2.718281828 6 index 0.5631606651 sub 0.17 div dup mul 6 index 0.79 sub 0.13 div dup mul add neg exp mul sub mul sub
  3 1 roll pop pop
  dup .88 mul neg .97 add exch
  dup .59 mul neg .96 add exch
  .49 mul neg .91 add
}
\pgfdeclarefunctionalshading{sfp-connected}{\pgfpoint{0bp}{0bp}}{\pgfpoint{100bp}{100bp}}{}{
  25 sub 50 div exch 25 sub 50 div exch
  1 1 2.718281828 3 index 0.22 0.07 360 8 index 0.64 sub mul cos mul add 0.05 720 8 index 0.64 sub mul cos mul add sub dup 0 lt { pop 0 } if 0.075 div dup mul neg exp sub 1 0.75 2.718281828 6 index 0.34 0.3 8 index 0.1579179607 sub 0.6320820393 div mul add 0.15 180 9 index 0.1579179607 sub 0.6320820393 div mul sin mul add sub 0.055 div dup mul 6 index 0.79 sub dup 0 lt { pop 0 } if 0.1 div dup mul add 0.1579179607 7 index sub dup 0 lt { pop 0 } if 0.1 div dup mul add neg exp mul sub mul sub
  3 1 roll pop pop
  dup .88 mul neg .97 add exch
  dup .59 mul neg .96 add exch
  .49 mul neg .91 add
}

\definecolor{C1}{RGB}{52,89,149}
\definecolor{C15}{RGB}{82,129,209}
\definecolor{C2}{RGB}{251,77,61}
\definecolor{C3}{RGB}{3,206,164}
\definecolor{C5}{RGB}{180,77,155}
\usepackage{hyperref}
\hypersetup{colorlinks=true,linkcolor=C5,citecolor=C5,urlcolor=C5}
\newtheorem{thm}{Theorem}

\theoremstyle{definition}
\newtheorem{defn}{Definition}
\theoremstyle{remark}

\newcommand{\mch}{\mathcal H}
\newcommand{\mcg}{\mathcal G}
\newcommand{\mcs}{\mathcal S}
\newcommand{\mcm}{\mathcal{M}}
\newcommand{\mco}{\mathcal O}
\newcommand{\mbc}{\mathbb C}

\newcommand{\id}{\mathds 1}
\newcommand{\mg}{\mathsf{MG}(n)}
\newcommand{\ad}{^{\dagger}}
\newcommand{\Span}{\operatorname{span}}
\newcommand{\eps}{\varepsilon}

\renewcommand{\leq}{\leqslant}
\renewcommand{\geq}{\geqslant}
\def\tn{^{\otimes n}}
\def\ts{^{\otimes 2}}
\def\Om{\Omega}
\def\Lm{\Lambda}
\def\ot{\otimes}

\title[Fermionic Gaussianity can be tested with mode-independent sample-complexity]{Fermionic Gaussianity can be tested with\\ mode-independent sample-complexity}

\author{\vspace{-5mm} Maxwell  West\,\textsuperscript{1,2}  and  Mart\'{i}n Larocca\,\textsuperscript{1,2}}
\address{\textsuperscript{1}Theoretical Division, Los Alamos National Laboratory, Los Alamos, New Mexico 87545, USA}
\address{\textsuperscript{2}Quantum Science Center, Oak Ridge, TN 37931, USA}

\date{}

\begin{document}
\begin{abstract}
Deciding whether an unknown quantum state either belongs to, or is far from, a given family of states is a natural question of quantum information theory. 
In many such cases, there is a natural 2-copy test which always accepts whenever the state indeed belongs to the target family; it is often much more difficult, however, to bound the probability of   acceptance for states that are far from the family.  
Here we develop a simple new technique for upper bounding this probability for general families, and therefore upper bounding the  sample-complexity of the decision problem itself.  As a particularly striking example of this framework, we show that deciding whether an unknown pure $n$-mode state is  fermionic Gaussian, or at least $\eps$-far in trace distance from all Gaussian   states, can be accomplished using the optimal number  $\Theta(\eps^{-2})$  of copies.  Remarkably, our analysis  applies to the well-known Bell sampling procedure for testing  fermionic Gaussianity, which we   therefore show to be optimal,   even compared to protocols making arbitrary collective and adaptive measurements. As a second example, we show  that the analogous decision problem for the family of Slater determinant states can also be solved with mode-independent sample-complexity.
\end{abstract}
\maketitle
\tableofcontents

\section{Introduction}
Quantum property testing~\cite{montanaro2013survey}, the testing of properties of quantum states~\cite{harrow2013testing,gross2021schur,soleimanifar2022testing,bittel2026complete}, unitary transformations~\cite{low2009learning,wang2011property}, and Hamiltonians~\cite{laborde2022quantum,bluhm2026hamiltonian}, is a major subfield within quantum information theory. As a prototypical example, one might have in mind a family of states and wonder, given access to copies of an unknown state, how efficiently   membership (or otherwise) of that family might be decided. In that spirit, of particular interest to us will be (pure)   fermionic Gaussian states, the states obtained by evolving a computational basis state under Hamiltonians quadratic in fermionic creation and annihilation operators, occupy a particularly useful position in quantum information
theory~\cite{valiant2001quantum,terhal2002classical,knill2001fermionic,bravyi2004lagrangian,jozsa2008matchgates}.
For example, their classical tractability has made them a basic reference point for understanding which additions to an otherwise classically simulable computation can supply a quantum advantage~\cite{jozsa2008matchgates,bravyi2006universal,de2013power}; more recently, the manifold of fermionic Gaussian states has played a key role in various  protocols for tomography~\cite{zhao2021fermionic,wan2022matchgate,mele2024efficient,chen2026optimal,west2026particle}, and for the benchmarking of quantum computers~\cite{helsen2022matchgate}. Given all of this interest, a particularly natural question in quantum property testing, then, is how difficult it is to decide  whether a given state is Gaussian or not.\\

Concretely, suppose that we are given access to independent copies of a fixed unknown pure state $\ket\psi$ on $n$ fermionic modes, and  are promised that $\ket\psi$ is either Gaussian or is at trace distance at least $\eps$ from every pure Gaussian state, and would like to decide which is true. How many copies do we need?
There has   been substantial progress towards answering this question in recent years~\cite{mele2024efficient,bittel2025optimal,lyu2024fermionic,walter2025random,tarabunga2026fermionic,haug2026practical}, with the best known sample-complexities falling steadily from  $\mco(n^5/\varepsilon^4)$~\cite{bittel2025optimal,mele2024efficient} to  $\mco(n^2/\varepsilon^2)$~\cite{tarabunga2026fermionic,walter2025random} to, very recently, $\mco(n/\varepsilon^2)$~\cite{haug2026practical}. The   algorithms  of several of these works~\cite{tarabunga2026fermionic,haug2026practical} rely on a remarkably simple \textit{Bell sampling} based test. Indeed,  with $c_1,\ldots,c_{2n}$ the Majoranas, one can show that (the tensor square of) every pure
Gaussian state is in the kernel of $\Lambda_n=\sum_{a=1}^{2n}c_a\otimes c_a$. As the operators $c_a\otimes c_a$ can be simultaneously measured (in the Bell basis) this readily leads   to a test which accepts for every Gaussian state: measure   their eigenvalues, and accept when they sum to zero. This does not immediately settle the testing problem, however, as one must also show that a state far from the Gaussian family has a noticeable probability of violating this criterion. To that end, Refs.~\cite{tarabunga2026fermionic,haug2026practical} have shown that this rejection probability  is large enough that their above-quoted sample-complexities suffice to decide the problem. In this work, we show that this $n$-dependence can be removed entirely:  the rejection probability $r_n(\psi)$ of the Bell test in fact satisfies $r_n(\psi)\geq \eps^2/8$, from which it follows that the sample-complexity is $\mco(\varepsilon^{-2})$. \\

Our argument, which allows us to analyse property testing algorithms in some generality, has a reasonably simple general form. We   consider   situations where we have some target family of states in mind, and a two-element projective POVM  $\{P,Q\}$ on two copies of the Hilbert space, such that $P$ accepts with certainty if given (the tensor square of) a state in the family, and would like to know about the rejection probability for states which are far from the family. Essentially, we derive some criteria under which we can prove that this rejection probability is bounded away from zero, allowing us to conclude that the corresponding property testing task can be solved with constant sample-complexity. As two examples of the framework, we  then consider the cases of fermionic Gaussianity testing, and the testing of whether a state is a Slater determinant. 
We begin by describing this setup in some generality, and then in subsequent sections specialise to  our examples. 

\section{The general case}\label{sec:gen}
We will assume that our target family forms a  compact  complex  submanifold $\mcm$ of the full manifold of states (for example, the orbits of highest-weight states of group irreps have such a complex structure~\cite{cavalli1986compact}). At a point   $\ket g$ on the target manifold, we denote its (projective) tangent space  by $T_g\mcm$. 
We assume that we have a two-element projective POVM $\{P, Q\}$ acting on two copies of the Hilbert space $\mch$ which satisfies \textit{perfect completeness}; that is, if $\ket g\in\mcm$, then the accepting probability $f(g):={\rm tr}\,[P (\ketbra{g})\ts]$ is equal to one. Our main preoccupation in this section will be to derive conditions under which we can say something non-trivial about the rejecting probability $r(\psi):={\rm tr}\,[Q (\ketbra{\psi})\ts]$ for states $\ket\psi$ which are far (in trace distance) from every element of $\mcm$.\\

To begin, we note that by the assumed   compactness of $\mcm$, there is a (possibly non-unique) state $\ket g\in\mcm$ which attains the  maximal fidelity $F_\mcm(\psi) = \max_{g'\in \mcm} \lvert\braket{g'}{\psi}\rvert^2$. Our first observation is that $\ket\psi$ is orthogonal to every vector in the tangent plane  $T_g\mcm$; that is, (up to possibly fixing some phases) we have 
\begin{equation}\label{eq:stationarity}
    \ket \psi=\sqrt {F_\mcm(\psi)} \ket g+\sqrt{1-F_\mcm(\psi)} \ket \chi,
\end{equation}
where $\ket\chi\in (\mathbb C\ket g\oplus T_g\mcm)^\perp$ (see Figure~\ref{fig:1}). To see this, take any vector $\ket u\in T_g\mcm$. We can then (by definition) find a curve $\gamma:[-\epsilon,\epsilon]\to \mcm$ with $\gamma(0)=\ket g,\ \dot{\gamma}(0)=\ket u$ for some $\epsilon>0$. But, by the assumed optimality of  $\ket g$, we then have
\begin{equation*}
    0 = \frac{d}{dt}\big[ \lvert \braket{\gamma(t)}{\psi}\rvert^2 \big] \Big\rvert_{t=0}= \frac{d}{dt}\big[  \braket{\gamma(t)}{\psi}\big] \braket{\psi} {\gamma(t)}  \Big\rvert_{t=0} +   \braket{\gamma(t)}{\psi} \frac{d}{dt}\big[\braket{\psi} {\gamma(t)} \big] \Big\rvert_{t=0} = 2\sqrt {F_\mcm(\psi)}\, \mathfrak{Re}  \braket{\psi} {u} ,
\end{equation*}
from which we see $\mathfrak{Re}  \braket{\psi} {u} =0$. Similarly, by the complexity of the manifold, the same argument applied to $iu$ allows us to conclude $\braket{\psi} {u} =0$.  In our examples, we will be interested in obtaining (via changes of variables) a decomposition of the form Eq.~\eqref{eq:stationarity} for which the state $\ket g$ has particularly nice properties. For example, in the fermionic Gaussian case we will see that we can without loss of generality take $\ket g$ to be the vacuum state, which will greatly facilitate the resulting analysis. \\

Now, we will always take  our POVM elements to act on two copies of the Hilbert space, so let us consider the tensor square of Eq.~\eqref{eq:stationarity}:
\begin{figure}
    \centering

\begin{tikzpicture}[
  x=4.25cm,y=3.4cm,font=\small,>=Stealth,
  sfp target/.style={draw=sfpInk,line width=.9pt},
  sfp point/.style={circle,fill=sfpInk,draw=white,line width=.3pt,
                    inner sep=0pt,minimum size=4.1pt}]
\def\sfpCurve{plot[domain=0:1,samples=300]
  (\x,{.22+.07*cos(360*(\x-.64))+.05*cos(720*(\x-.64))})}

% (a) A tangent line represents the tangent space in this schematic.
\begin{scope}
  % \node[anchor=west] at (0,1.085) {(a)\enspace Geometry};
  \draw[sfp target] \sfpCurve;
  \node[anchor=south west,text=sfpInk] at (.055,.24) {$\mathcal M$};
  \draw[gray,densely dashed,line width=.55pt] (.31,.34)--(.94,.34);
  \node[anchor=south east,text=gray] at (.96,.355) {$T_g\mathcal M$};
  \draw[sfpInk,line width=.65pt] (.64,.34)--(.64,.79);
  \draw[sfpInk,line width=.5pt] (.64,.389)--(.679,.389)--(.679,.34);
  \node[sfp point] at (.64,.34) {};
  \node[anchor=north east,inner sep=3pt] at (.685,.325) {$\ket g$};
  \node[sfp point, fill=C15] at (.64,.79) {};
  \node[anchor=south west,inner sep=3pt] at (.15,.82) {\textcolor{C1}{$\ket\psi = \sqrt{F} \ket g + \sqrt{1-F}\ket \chi$}};
\end{scope}

% (b) The superlevel component through psi misses the target.
\begin{scope}[shift={(1.25,0)}]
  % \node[anchor=west] at (0,1.085) {(b)\enspace Ruled out};
  \shade[shading=sfp-island, rounded corners=7pt] (0,0) rectangle (1,1);
  \draw[sfp target] \sfpCurve;
  \node[sfp point] at (.64,.79) {};
  \node[anchor=south west,inner sep=3pt] at (.64,.79) {$\ket\psi$};
  \draw[rounded corners=7pt] (0,0) rectangle (1,1.);
  
\end{scope}

% (c) The coloured corridor indicates a path in the starting superlevel set.
\begin{scope}[shift={(2.4,0)}]
  % \node[anchor=west] at (0,1.085) {(c)\enspace A path exists};
  \shade[shading=sfp-connected, rounded corners=7pt] (0,0) rectangle (1,1);
  \draw[sfp target] \sfpCurve;
  \node[sfp point] at (.64,.79) {};
  \node[anchor=south west,inner sep=3pt] at (.64,.79) {$\ket\psi$};
  \node[anchor=west,inner sep=2pt] at (.735,.52) {$\gamma$};

  \draw[rounded corners=7pt] (0,0) rectangle (1,1.);
  
\end{scope}

% One shared scale; no axes or extra annotations in the heat maps.
\shade[bottom color=sfpLow,top color=sfpHigh]
  (3.565,.15) rectangle (3.603,.85);
\node[anchor=west,inner sep=3pt] at (3.603,.15) {$0$};
\node[anchor=west,inner sep=3pt] at (3.603,.85) {$1$};
\node at (3.585,.965) {$f$};

\end{tikzpicture}
\setlength{\captionindent}{5pt}
    \caption{In general, we consider the sample-complexity of the task of deciding whether a state $\ket\psi$ belongs to a submanifold $\mcm$ of states, or is at least $\varepsilon$-far in trace distance from all states in $\mcm$. It is helpful to decompose it into a nearest state $\ket g\in\mcm$, along with a component orthogonal to the tangent plane of $\mcm$ at $\ket g$ (left panel). We will work with property testers whose acceptance probability $f(\ket\psi)$ is always perfect for states within $\mcm$, and are interested in deriving conditions under which we can prove that they do not also erroneously accept states far from $\mcm$ (middle panel). To that end, we investigate settings under which states can be proved to be deformable to the target manifold along paths whose acceptance probability does not fall below that of the initial state  (right panel). Under such a condition (which we show holds for both fermionic Gaussianity testing and Slater determinant testing), we are able to show that states far from $\mcm$ are accepted with probability bounded away from one, which allows us to establish constant sample-complexity scaling. }
    \label{fig:1}
\end{figure}
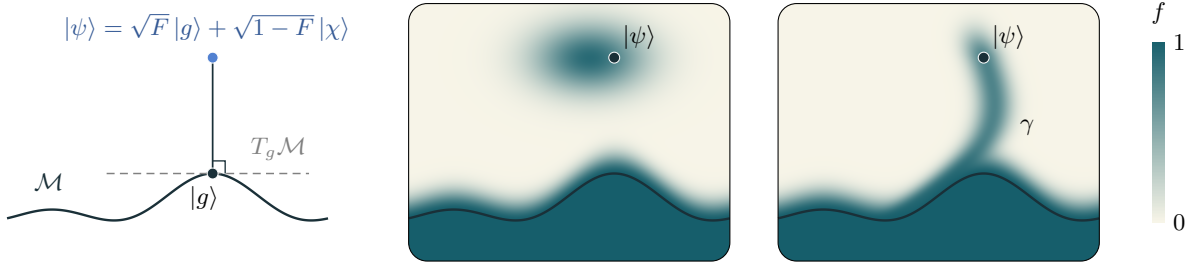
\begin{equation}\label{eq:squared}
\ket \psi^{\otimes2} =F_{\mcm}(\psi)\ket g^{\otimes2}+\sqrt{F_{\mcm}(\psi)(1-F_{\mcm}(\psi))}(\ket g \ket \chi+\ket \chi  \ket g) +(1-F_{\mcm}(\psi))\ket \chi^{\otimes2}.
\end{equation}
The first term of Eq.~\eqref{eq:squared} is (by assumption) always accepted by $P$, so that controlling the rejection probability for states beyond the target family is entirely a question of the final two terms. Roughly speaking, our hope is that, for states not in the target family, the middle and final terms do not ``cancel too much'', which would lead to that state incorrectly passing the test. One way to get good control over this comes from the hypothesis of the following lemma:
\begin{restatable}{lem}{lemwitness} \label{lem:witness}
Suppose that for every target $\ket g\in\mcm $ and   $\ket \chi\in ( \ket g\oplus T_g)^\perp$ there is a  
$\ket w \in\operatorname{ran}Q$ such that $\langle w,\chi^{\otimes2}\rangle=0$, $|\langle w,g\otimes\chi+\chi\otimes g\rangle|\ge a$, and $\| w\|^2\le B$, where $a$ and $B$ do not depend on $\ket g$ or $\ket\chi$. Then, with $c=a^2/B$,  we have $r(\psi)\ge cF_\mcm(\psi)(1-F_\mcm(\psi))$.
\end{restatable}
\begin{proof}
As $w$ is in the range of the projector $Q$, we have $r(\psi) = \bra{\psi}\ts Q \ket\psi\ts \ge \bra{\psi}\ts (\ketbra{w}{w}) \ket\psi\ts/\|w\|^2$. Meanwhile, 
by the various hypotheses of the lemma, the first and last terms on the right hand side of Eq.\eqref{eq:squared} are orthogonal to  $w$, so that   $|\langle w,\psi^{\otimes2}\rangle|\ge a\sqrt{F(1-F)}$, from which the result follows. 
\end{proof}
For example, in the Gaussian and Slater determinant sections we will find such ``witness vectors'' $\ket w$ with corresponding values of $c=1/2$ and $c=2/3$, respectively. Now, as we are interested in lower bounding the rejection probability for states outside of the target family, Lemma~\ref{lem:witness} is a good start. One unfortunate property of its bound, however, is that $cF(1-F)$ vanishes  at $F=0$. That is to say, for states with very low fidelity overlap with the target family, the bound of Lemma~\ref{lem:witness} is vacuous (which is unfortunate because those are exactly the states which we would like to be rejecting very strongly). The remainder of this section, essentially, is then about finding conditions under which we can strengthen our bound to also hold for states for which $F$ is very small. To that end, we have a definition:
\begin{defn}\label{def:path}
A family with test acceptance probability function $f$ has the \textit{path property} if every   $\ket\psi$ can be joined continuously to the family  by a path $\gamma:[0,1]\to \mch$ of unit vectors satisfying $\gamma(0)=\ket\psi,\  \gamma(1)\in\mathcal M,$ and $f(\gamma(t))\ge f(\psi) $ for $ 0\le t\le1$.
\end{defn}
We emphasise that Definition~\ref{def:path} allows the acceptance probability can rise and fall along the path, provided it never falls below its original value. It turns out that the path property combined with Lemma~\ref{lem:witness} gives us a very strong sample-complexity bound:
\begin{thm}\label{thm:1}
If $\mathcal M$ is compact and the path property and Lemma~\ref{lem:witness} hold, then $r(\psi)\ge  c(1-F_\mcm(\psi))/4$ for all pure states $\ket\psi$. This implies that, for $0<\varepsilon,\delta<1$, independently repeating the test $ \mco(  c^{-1}\varepsilon^{-2}\log 1/ \delta)$ times succeeds at   deciding whether a state is with $\mcm$ or at trace distance at least $\varepsilon$   with probability   at least $1-\delta$.
\end{thm}
\begin{proof}
Take some state $\ket\psi$. If  $F_\mcm(\psi)\ge1/2$, then Lemma~\ref{lem:witness} immediately implies the claimed   bound. If  $F_\mcm(\psi)<1/2$, let us   choose a path given by the path property. Now, the function $F_\mcm$ is continuous, since it is the maximum of a continuous function over the compact target. As it begins below $1/2$ and ends at one,  by the intermediate value theorem it equals $1/2$ at some point. At that point Lemma~\ref{lem:witness} shows that  the rejection probability is at least $c/4$;  by the path property, this is also at most the starting rejection. This proves that, either way, $r(\psi)\ge  c(1-F_\mcm(\psi))/4$.  Now, for $\varepsilon$-far-from $\mcm$ states we have $1-F_\mcm(\psi)\geq \varepsilon^2$, so that   every faraway state rejects a given round of the test with probability at least $c\varepsilon^2/4$. Its probability of passing $R$
independent rounds is then at most $\exp(-Rc\varepsilon^2/4)$, from the second claim of the theorem follows.
\end{proof}

Now, admittedly, it is not entirely clear that one would expect the path property of Definition~\ref{def:path} to be true particularly often in  examples. The way we will realise it in practice is through a sort of recursive procedure.  To see this, suppose we have a decomposition $\mch\cong \mch_0\oplus \mch_1$, and that our accepting POVM element $P$ commutes with (say) $T\ts$, where  $T=\id_{\mch_0}\oplus (i\id_{\mch_1})$. We can always write a given state  $\ket\psi\in\mch$ as $\ket\psi = \sqrt{t_0}\ket x+\sqrt{1-t_0}\ket y$ for some $\ket x\in \mch_0,\ \ket y\in \mch_1$, so that $\ket\psi\ts = t_0\ket x\ts + \sqrt{t_0(1-t_0)}(\ket x\ket y+\ket y\ket x) + (1-t_0)\ket y\ts$; more generally,   let us consider the parameterised family $\ket{\psi_t} = \sqrt{t}\ket x+\sqrt{1-t}\ket y$ for $0\leq t\leq 1$. The usefulness behind this setup is that, by the assumption of the commutativity  of $P$ with $T\ts$, most of the cross terms in the expression $f(\ket{\psi_t})={\rm tr}\,[P (\ketbra{\psi})\ts]$ for the acceptance probability are going to vanish, leading to
\begin{equation*}
    f(\ket{\psi_t}) = \bra{\psi_t}\ts P \ket{\psi_t} \ts = At^2 + Bt(1-t) + C(1-t)^2
\end{equation*}
for some $A,B,C$ independent of $t$. Now, here comes the mechanism responsible for getting the path property to work in practice: As the acceptance probability is a quadratic in $t$, if we could somehow guarantee that $\max\{f(\ket{\psi_0}),f(\ket{\psi_1})\}\geq f(\ket{\psi_{t_0}})$, then we could clearly deform $\ket{\psi_{t_0}}$ to one of those  ``endpoints'' without ever lowering its acceptance probability throughout the deformation. Having done that, we would have a state that lives completely in  either $\mch_0$ or $\mch_1$;   we   then seek to recursively apply such a procedure, ultimately ending up with a simple state that, if we set things up correctly, lives in $\mcm$. For example, in the fermionic Gaussianity case, the decomposition will be in terms of splitting the Hilbert space into the occupied / unoccupied sectors of a given mode. Inductively repeating this on the remaining $(n-1)$-mode system will ultimately lead us to a single computational basis state, which is fermionic Gaussian. \\

\noindent
All of this is perhaps made a little clearer through an explicit example, to which we now turn.

\section{Fermionic Gaussianity testing}\label{sec:fg}
As advertised, our primary application is to fermionic Gaussianity testing, which we will show can be accomplished with failure probability at most $\delta$ using the optimal number $\mco(\eps^{-2}\log(1/\delta))$   of samples. We begin by recalling some basic properties of the setup.
Let $\mch_n\cong(\mbc^2)^{\otimes n}$ be the Fock space of $n$ fermionic modes, with computational  basis states denoting which modes are occupied. We write $\ket0$ for the vacuum and make the standard Jordan--Wigner choice
\begin{equation}\label{eq:majorana}
 c_{2j-1}=Z_1\cdots Z_{j-1}X_j,\qquad
 c_{2j}=Z_1\cdots Z_{j-1}Y_j,\qquad 1\leq j\leq n.
\end{equation}
of Majoranas.   Equivalently, writing $a_j$ for the
annihilation operator of the $j$\textsuperscript{th} mode, we have $c_{2j-1}=a_j+a_j\ad$ and
$c_{2j}=-i(a_j-a_j\ad)$. The \textit{matchgate group}~\cite{jozsa2008matchgates} $\mg\cong\operatorname{Spin}(2n)$ is then generated by 2-Majorana rotations
\begin{equation}\label{eq:rotations}
 R_{ab}(\vartheta)=\exp\!\left(-\frac\vartheta2c_ac_b\right),
 \qquad a<b.
\end{equation}
Every $U\in\mg$ acts on the Majoranas through an orthogonal matrix, $Uc_aU\ad=\sum_b O_{ab}c_b$, with $O\in\operatorname{SO}(2n)$.
These unitaries preserve the fermionic parity operator
$\Pi=Z_1\cdots Z_n$, whose $\pm1$ eigenspaces are denoted
$\mch_{n,\pm}$. We define the even and odd parity pure Gaussian families by
\begin{equation}\label{eq:family}
 \mcg_{n,+}=\{U\ket0:U\in\mg\},\qquad
 \mcg_{n,-}=\{Uc_1\ket0:U\in\mg\},
\end{equation}
and put $\mcg_n=\mcg_{n,+}\cup\mcg_{n,-}$.  For a unit vector $\ket\psi$, we define its Gaussian fidelity and distance to the Gaussian manifold respectively by
\begin{align}\label{eq:distance}
 F_{\mcg}(\psi)&=\max_{g\in\mcg_n}\ \lvert\braket{g}{ \psi}\rvert^2,\\
 d_{\mcg}(\psi)&=\sqrt{1-F_{\mcg}(\psi)}.
\end{align}

In the testing fermionic Gaussianity   problem, the tester receives independent copies of   some unknown pure state $\ket\psi$, which is promised either to belong to $\mcg_n$ or to satisfy $d_{\mcg}(\psi)\geq\eps$.  Our analysis will focus on the Bell sampling procedure of Refs~\cite{tarabunga2026fermionic,haug2026practical}; namely, with   $\Lambda_n=\sum_a c_a\ot c_a$, we consider the POVM with elements given by $P_n={\rm proj}_{\ker \Lm_n}$, the projector onto the kernel of $\Lm_n$, and $Q_n=\id-P_n$. We will use the notation $ f_n(\psi)=\bra\psi^{\otimes2}P_n\ket\psi^{\otimes2}$, and $r_n(\psi)=1-f_n(\psi)$ for the acceptance and rejection probabilities respectively. One can readily verify by a direct calculation that $\Lm_n$ annihilates the (tensor square of the) vacuum, and furthermore that $(U\otimes U)\Lambda_n(U\ad\otimes U\ad)=\Lambda_n$ for all $U\in\mg$, from which it follows immediately that the Bell sampling test has perfect completeness, accepting $\ket\psi\ts$ with probability one whenever $\ket\psi$ is fermionic Gaussian. Although it is not particularly important for our purposes, we remark that the name ``Bell sampling'' comes from the fact that the basis in which the commuting operators $c_a\ot c_a$ are simultaneously diagonalised  is exactly a Bell basis, so that what one can do is simultaneously measure all of their eigenvalues, and then accept if and only if the sum of those eigenvalues is zero~\cite{tarabunga2026fermionic,haug2026practical}. Let us show that the sample-complexity of this procedure is optimal, which implies:

\begin{thm}\label{thm:main}
Pure $n$-mode Gaussian states can be distinguished from pure states satisfying $d_{\mcg}(\psi)\geq\eps$ with failure probability at most $\delta$  using the optimal number $\Theta(\varepsilon^{-2}\log(1/\delta))$ of copies, for all $0<\varepsilon,\delta<1/2$.
\end{thm}
\begin{proof}
We will show that repeating the Bell test   $R=\left\lceil 8\eps^{-2}\log 1/\delta\right\rceil$ times, and accepting if and only if every round accepts, distinguishes pure $n$-mode Gaussian states from pure states satisfying
$d_{\mcg}(\psi)\geq\eps$, with perfect completeness (i.e., all Gaussian states are accepted) and soundness error
at most $\delta$ (i.e., all $\varepsilon$-far-from-Gaussian states are accepted with probability at most $\delta$). 
At a high level, the proof has three steps, which   follow along the general construction from Section~\ref{sec:gen}:

\begin{enumerate}[label=(\arabic*),leftmargin=*]
 \item  By considering the nearest Gaussian state to $\ket\psi$ (which we will see  without loss of generality may be taken to be the vacuum),  we show that $r_n(\psi)\geq F_{\mcg}(\psi)(1-F_{\mcg}(\psi))/2$ (cf. Lemma~\ref{lem:witness}).
 \item Next we show that the initial state can be continuously deformed to a Gaussian state, such that at all points along the path the Bell test acceptance never falls below the initial value $f_n(\psi)$ (cf. Definition~\ref{def:path}).
 \item Finally, we    appeal to  Theorem~\ref{thm:1}.
\end{enumerate}

We emphasise that these  are only steps of the analysis, and the actual experimental procedure is just the usual Bell test (and so involves neither finding a closest Gaussian state nor finding a suitable deformation path); we are just showing that the sample-complexity of this existing procedure is better than was previously known.\\

\textit{Step 1.}
For simplicity, let us begin by assuming that  $\ket{\psi'}$ has definite parity, which we without loss of generality take to be even (if it is odd we instead apply our results to  $c_1\ket{\psi'}$; as $c_1\ts \Lm_n (c_1\ad)\ts=\Lm_n$, all of our results go through the same) and choose a Gaussian $\ket g$ which maximises the  overlap with $\ket{\psi'}$. Our first observation is that,  by applying to $\ket{\psi'}$ the inverse of a Gaussian unitary $U$ which prepares $\ket g$, we may without loss of generality take the closest Gaussian state to be $\ket0$; indeed, deciding whether $\ket\psi:=U\ad\ket{\psi'}$ is Gaussian or far-from-Gaussian is clearly equivalent to deciding the same question with respect to $\ket{\psi'}$ itself. So, we can write $\ket\psi = \alpha\ket 0 + \sqrt{1-\alpha^2}\ket\chi$ for some $\ket\chi$ which is orthogonal to the tangent plane $T_{\ket 0}\mcm$ of the manifold $\mcm$ of Gaussian states at the vacuum $\ket 0$. Now, it turns out that being orthogonal to $T_{\ket 0}\mcm$ translates to some simple conditions on the expansion of $\ket\chi$ in the computational basis (whence the usefulness of considering $U\ad\ket{\psi'}$ in the first place). The point is that  an infinitesimal Gaussian motion away from the vacuum creates a \textit{pair} of fermions (which follows from considering the action of (say) $\exp(\epsilon(a_i\ad a_j\ad+a_i a_j\ad+a_i\ad a_j+a_i a_j))$ on the vacuum, which is annihilated by any annihilation operator). Thus, if $\ket\psi$ had a nonzero two-particle amplitude, a Gaussian slightly in that direction from the vacuum would have an increased   overlap with $\ket\psi$, contradicting the optimality of the vacuum. More formally, we find (recalling that at this point we are assuming $\ket\psi$ to be of even parity):
\begin{restatable}{lem}{lemstationarity}\label{lem:stationarity}
With $F=F_{\mcg}(\psi)>0$, we have $\ket\psi=\sqrt F\ket0+\sum_{\substack{|K|\geq4}}\psi_K\ket K$, where $|K|$ is the Hamming weight of the computational basis string $K$.
\end{restatable}
We postpone the (short) calculation which rigorously establishes Lemma~\ref{lem:stationarity}  to
Appendix~\ref{app:fm}; the point is just that, as alluded to above, there are no two-particle excitations (and as $\ket\psi$ is even, there are by definition no one- or three-particle excitations either). The vanishing of the two-particle excitations turns out to be particularly convenient for reasoning about the   acceptance probability of the Bell test; indeed, one can exploit it to find:

\begin{restatable}{lem}{lemnormal}\label{lem:normal}
If an even-parity state has vacuum probability $p$ and zero two-particle amplitudes, then the rejection probability of the Bell test satisfies $r_n(\psi)\geq {p(1-p)}/2.$
Consequently, for every pure state of definite parity,
$r_n(\psi)\geq F_{\mcg}(\psi)(1-F_{\mcg}(\psi))/2$.
\end{restatable}

Lemma~\ref{lem:normal} provides a bound on the rejection probability exactly of the form of Lemma~\ref{lem:witness}; the way we prove it (in   Appendix~\ref{app:fm}) is to use the structure of $\ket\psi$ guaranteed by Lemma~\ref{lem:stationarity} to deduce the existence of the sort of ``witness vector'' required by Lemma~\ref{lem:witness}.
This concludes the first step.

\medskip\noindent
\textit{Step 2.}
Briefly stepping back, recall that the Bell test passes with certainty if the state is indeed Gaussian, and we are trying to lower bound the test rejection probability $r_n(\psi)$ for far-from-Gaussian states (that is, states $\ket\psi$ with $F_{\mcg}(\psi)\leq1-\varepsilon^2$). Now, if  $F_{\mcg}$ is additionally bounded away from zero, Lemma~\ref{lem:normal} will give us exactly such a lower bound. Unfortunately, for $F_{\mcg}\in o(1)$ (i.e., exactly the states which we would like to be rejecting with prejudice) the bound of Lemma~\ref{lem:normal} becomes vacuous. Our remaining task, then, is    to show that these states are indeed rejected as they should be; this is  the point where we will attempt to apply the path argument of Section~\ref{sec:gen}. \\ 

To that end, it is helpful to introduce, for $s\in\{+1,-1\}$, the projectors $ \Pi_{ab,s}=({\id+isc_ac_b})/2$ onto the occupied and unoccupied states corresponding to the mode associated to $c_ac_b$, and the overlap $a_{ab,s}=\|{\Pi_{ab,s}\ket \psi}\|^2$ of $\ket\psi$ with those states. For example, if $(a,b)=(2j-1,2j)$ for some integer $j$, then we have $ \Pi_{(2j-1,2j),s}=({\id-sZ_j})/2$, the usual projectors onto the  occupied / unoccupied states of the physical mode $j$. More generally, there is some matchgate $V$ such that $Vc_aV^\dagger=c_1$, and $Vc_bV^\dagger=c_2$, so that
\begin{equation*}
V\Pi_{ab,s}V^\dagger =\frac{\id-sZ_1}{2}=\ketbra{q_s}\otimes \id,
\end{equation*}
where $q_s=(1+s)/2.$
Now, we can always write $|\psi\rangle = V\ad \big( |0\rangle\otimes|\xi_0\rangle+|1\rangle\otimes|\xi_1\rangle\big)$ 
for some (unnormalised) $\ket{\xi_0},\ket{\xi_1}$; projecting then selects one of these two summands, so that 
\begin{equation}\label{eq:abs}
\frac{V\Pi_{ab,s}|\psi\rangle}{\sqrt{a_{ab,s}}}=|q_s\rangle\otimes|\phi_{ab,s}\rangle,
\end{equation}
where $\ket{\phi_{ab,s}}$ is some $(n-1)$-mode (probably non-Gaussian) state of definite parity (indeed, by assumption $\ket\psi$ is of definite parity, the matchgate $V$ and the projector $\Pi_{ab,s}$ preserve parity, and $\ket{q_s}$ is of definite parity, which collectively forces $\ket{\phi_{ab,s}}$ to also be of definite parity).
Now, the point of this diversion is that we can also ask about the acceptance probability $f_{n-1}(\phi_{ab,s})=\tr\big[P_{n-1}\ketbra{\phi_{ab,s}}{\phi_{ab,s}}\ts\big]$ of $\ket{\phi_{ab,s}}$ under the $(n-1)$-mode Bell test; we find:

\begin{restatable}{lem}{lembranching}\label{lem:branching}
For every $n\geq2$ and every pure $n$-mode state $|\psi\rangle$ of definite parity, there exist $a<b$ and $s\in\{+1,-1\}$, with $a_{ab,s}>0$, such that $f_{n-1}(\phi_{ab,s})\geq f_n(\psi).$
\end{restatable}
That is, some single-mode occupation measurement has an outcome whose conditional state passes the smaller test at least as often as the original state passes the original test.  \\

To see how this helps with the low-fidelity states that defy the bound of Lemma~\ref{lem:normal}, suppose we could continuously deform our original state into a Gaussian, while never lowering its acceptance probability below its original value (we are for clarity essentially repeating some arguments from Section~\ref{sec:gen} here). If the original Gaussian fidelity were below $1/2$, then by the intermediate value theorem somewhere along this deformation it would have to equal $1/2$. At that point, Lemma $\ref{lem:normal}$ would force a rejection probability of at least $1/8$. Since acceptance along the deformation (by the speculative assumption of this paragraph) never falls below its initial value, the original state would have to reject with probability at least $1/8$ as well. Thus, constructing such a path would give precisely the missing bound.
Lemma $\ref{lem:branching}$ now gives some intuition about how one might proceed! By choosing a suitable tuple $(a,b,s)$, we can transform our state into one with  a definite occupation number at that     mode while preserving or increasing acceptance; repeating this on the remaining modes eventually leaves a (Gaussian) occupation state. The difficulty is that these projections make discrete jumps: the Gaussian fidelity could jump from below $1/2$ to above it without ever taking the intermediate value (and note that the bound of Lemma~\ref{lem:normal} is also weak for very high fidelity states). We would therefore like to show that each conditioning step can instead be reached continuously, without allowing acceptance to fall below its starting value.\\

So, let us choose a pair of Majoranas for which Lemma~\ref{lem:branching} supplies an improving outcome, and make the Gaussian change of coordinates $V$ described above. We will interpolate between its two occupation branches by varying only their relative weights. Concretely, let us write, for $0\leq t\leq1$,
\begin{equation}\label{eq:interpolation}
 \ket{\psi_t}=V\ad\big(\sqrt t\ket0\ket x+\sqrt{1-t}\ket1\ket y\big),
\end{equation}
where $x$ and $y$ are normalised and of opposite parities, and the original state occurs at some $t=t_0$. Now, we will see in Appendix~\ref{app:fm} that $P_n$ commutes with the total occupation across the two copies of (any) mode; it follows that we have 
\begin{align}
f_n(\psi_t)&= \bra{\psi_t}\ts P_n \ket{\psi_t}\ts =At^2+Bt(1-t)+C(1-t)^2\label{eq:quadratic}
\end{align}
where $A=f_{n-1}(x),\  C=f_{n-1}(y)$, and $B$ is the expectation of $P_n$ on the (unnormalised) cross term. 
So, the acceptance probability is a quadratic in $t$,   which we know to satisfy $\max\{f_n(\psi_0),f_n(\psi_1)\}\geq f_n(\psi_{t_0})$. But, for  a segment of a quadratic having an endpoint at least as high as its starting value, it is easy to see that there must be  a route to (at least one of the) endpoints which never goes below that starting value.  We can then recurse the procedure, which immediately yields:
\begin{restatable}{lem}{lempath}\label{lem:path}
For every pure state $\ket\psi$ of definite parity there is a continuous path of unit vectors $\gamma:[0,1]\to\mch_n$ of the same parity such that $\gamma(0)=\ket\psi,\ \gamma(1)\in\mcg_n$, and $f_n(\gamma(t))\geq f_n(\psi)$ for all $t\in[0,1]$.
\end{restatable}
\noindent
This concludes the second step.

\medskip\noindent
\textit{Step 3.}
We can now put the two ingredients together (this is basically just Theorem~\ref{thm:1}, but we spell it out). If the original state has $F\geq1/2$, Lemma~\ref{lem:normal} already gives
$r_n\geq F(1-F)/2\geq (1-F)/4\geq (1-F)/8$. If instead $F<1/2$, take the path supplied by Lemma~\ref{lem:path}. Gaussian fidelity is continuous, since it is a maximum over a compact family. It starts below $1/2$ and ends at one, so some state on the path has Gaussian fidelity exactly $1/2$. By Lemma~\ref{lem:normal}, that state rejects with probability at least $1/8$. By Lemma~\ref{lem:path}, its rejection probability is no greater than the initial rejection, which is in particular more than $(1-F)/8$. Either way, we therefore have $r_n(\psi) \geq(1-F)/8\geq \varepsilon^2/8$, so that the probability of accepting all $R$ independent rounds is therefore at most $(1-\eps^2/8)^R\leq e^{-R\eps^2/8}\leq\delta$.\\

At this point, we have established the claim of the theorem under the assumption that our state is of definite parity. More generally, let us write an arbitrary pure state as $\ket\psi=\sqrt p\ket{\psi_+}+\sqrt q\ket{\psi_-}$, with $p+q=1$,  where $\psi_\pm$ are normalised states of definite even and odd  parity, and put $F_\pm=F_{\mcg}(\psi_\pm)$. Since Gaussian states have definite parity, we then have $F_{\mcg}(\psi)=\max\{pF_+,qF_-\}.$ Now, from the readily seen fact that  $P_n$ commutes with the single-copy  parity operators $\Pi\ot\id$ and $\id\ot\Pi$, and additionally vanishes on the mixed-parity sectors, it follows that most of the cross terms in $\bra{\psi}\ts P_n \ket{\psi}\ts$ are going to vanish, leading to $ f_n(\psi)=p^2f_n(\psi_+)+q^2f_n(\psi_-)$, and, using $p+q=1$, that  $r_n(\psi)=p^2r_n(\psi_+)+q^2r_n(\psi_-)+2pq.$
The preceding  definite-parity case analysis then gives $r_n(\psi_\pm) \geq(1-F_\pm)/8$, whence
\begin{align*}
 8r_n(\psi)\geq p^2(1-F_+)+q^2(1-F_-)+16pq=1+14pq-p^2F_+-q^2F_-\geq1-F_{\mcg}(\psi),
\end{align*}
where the final step used $p^2F_++q^2F_-\leq(p+q)\max\{pF_+,qF_-\}$. This immediately implies the statement of the theorem in the general case. \\

Finally, the fact that our $\varepsilon$ and $\delta$ dependence are optimal comes from comparing the vacuum state with $\sqrt{1-\eps^2}\ket{0}\tn+\eps\ket{1}\tn$ (take $\varepsilon<1/2$ and $n>2$); one can show~\cite{haug2026practical} that its distance to the Gaussian manifold is $\varepsilon$; meanwhile it cannot be distinguished from $\ket0 \tn$ with fewer than $\Om(\varepsilon^{-2}\log(1/\delta))$ copies~\cite{watrous2018thetheory}. 
\end{proof}

\section{Slater determinant testing}\label{sec:sd}
We now turn to our next example,   that of Slater determinant testing. As in the previous example, we will   follow the recipe of Section~\ref{sec:gen}.  We  consider  $\eta$-particle states of an $n$-mode fermionic system,  the relevant Hilbert space for which is $\mch_\eta=\Lm^{\eta}(\mbc^n)$. Given a basis $e_1,e_2,\ldots,e_n$ of $\mbc^n$, we recall that a  Slater determinant $\ket \phi$ is given by the exterior action $\Lm^\eta(U)$ of a $U\in {\rm U}(n)$ on one of the occupation-basis vectors $\ket Z  = e_{z_1} \wedge \cdots \wedge  e_{z_\eta} $, where $Z=\{z_1<\cdots<z_\eta\}\subset[n]$ (without loss of generality we can take $Z=\{1,2,\ldots,\eta\}$, absorbing everything into $u$). Our target family of Slater determinant states (let us call it $\mcs_{n,\eta}$), then, is parameterised by all such choices of $U\in {\rm U}(n)$, and we will often identify such a state with a corresponding $U$. The first step is to decide which operators to use in our two-element two-copy POVM $\{P,Q\}$. We will find it  simplifies the resulting computations somewhat   to reframe the canonical projector onto the span of the tensor square of the Slater determinants 
as the projector $P={\rm proj}_{\ker(T_{n,\eta}\ad T_{n,\eta})}$, where $T_{n,\eta}=\sum_{i=1}^n a_i\ot a_i\ad$ (we will also  sometimes use the notation $L_{n,\eta}=T_{n,\eta}\ad T_{n,\eta}$). Indeed, it is immediate that this test will accept   any Slater determinant state: up to a rotation by $\Lm^\eta(V)$ for some $V\in{\rm U}(n)$ (which commutes past $T_{n,\eta}$) a given mode is either occupied or unoccupied; either way, for each $i$, one of $a_i$ or $a_i\ad$ will annihilate $\ket\phi\ts$ if $\ket\phi$ is a  Slater determinant.  As in the case of fermionic Gaussianity testing, the remaining task is then to use the ideas of Section~\ref{sec:gen} to bound the rejection probability for states that are far from the accepted manifold. We find:

\begin{thm}\label{thm:slatmain}
Pure $n$-mode $\eta$-particle Slater determinant states can be distinguished from pure states satisfying $d_{\mcs_{n,\eta}}(\psi)\geq\eps$ with failure probability at most $\delta$  using the optimal number $\Theta(\varepsilon^{-2}\log(1/\delta))$ of copies, for all $0<\varepsilon,\delta<1/2$.
\end{thm}
\begin{proof}
Let us dutifully follow our recipe. The first step is to find a witness vector (cf. Lemma~\ref{lem:witness}) so as to establish a result analogous to Lemma~\ref{lem:normal}. Again, the key is to bring our state to the form of Eq.~\eqref{eq:stationarity} for convenient choice of $\ket g$, and then build a witness   vector which is informed by the structure of the component   orthogonal to the tangent plane to the manifold of target states.    In the Slater case we find it convenient to take $\ket g$ to be given by $e_1\wedge e_2\wedge \ldots \wedge e_\eta$, some analysis of which leads (in a very similar way to the fermionic Gaussianity example) to:

\begin{restatable}{lem}{lemslatnormal}\label{lem:slatnormal}
For every pure state $\ket\psi$, the rejection probability of our Slater determinant test satisfies $r_{n,\eta}(\psi)\geq (2/3)F_{\mcs_{n,\eta}}(\psi)(1-F_{\mcs_{n,\eta}}(\psi))$.
\end{restatable}

The next step is to seek a way to connect an arbitrary state to the manifold of Slater determinant states along a path satisfying the criteria of Definition~\ref{def:path}. In the fermionic Gaussian case we approached this by deforming, one mode at a time, into a state for which   the mode currently under   consideration was either exactly occupied or unoccupied; we will do essentially the same thing here. Indeed, take some single particle orbital $u\in\mbc^n$. The $\eta$-particle Hilbert space decomposes as
\begin{equation}\label{eq:sd_split}
 \mch_{n,\eta}=\Big(\ket u\wedge\bigwedge^{\eta-1}\ket {u}^\perp\Big)\oplus\bigwedge^\eta \ket u^\perp,
\end{equation}
into subspaces where the $u$ orbital is occupied or not; indeed, on the left hand side of the right hand side of Eq.~\eqref{eq:sd_split}  one of the occupied states is exactly $u$ and the other $\eta-1$ occupied states are orthogonal to it, and on the right all of the occupied states orthogonal to it. Now, recall from our analysis of this step in the general case (Section~\ref{sec:gen}) that we would like our accepting POVM element to commute with (the tensor square of) an operator which is diagonal with respect to the decomposition of Eq.~\eqref{eq:sd_split}, and is proportional to the identity on both blocks (but with different eigenvalues in the two cases).  In this case, that can be arranged by considering the occupation number of the $u$-mode; that is, with $a(u)\ad:=\sum_i u_i a_i\ad$, the operator $N_u = a(u)\ad a (u) $. Then, note that $N_u\ot\id +\id\ot N_u$ (which counts the total occupancy across the two copies of the mode $u$, and generates the tensor square of $\exp i \theta N_u$) commutes with our POVM elements. Indeed, $T_{n,\eta}$ preserves the total occupancy across the two copies of the mode $u$, because it simply swaps occupancies between the two copies, and so commutes with $N_u\ot\id +\id\ot N_u$; it follows that $P_{n,\eta}$ also commutes with it, as does $L_{n,\eta}=T_{n,\eta}\ad T_{n,\eta}$. \\

Now, the two-copy Hilbert space decomposes into sectors of $u$-mode occupancy zero, one, and two in a manner that respects the acceptance probability of our test. Concretely, 
let us use   $J_1:x\mapsto u\wedge x$ and $J_0:y\mapsto y$ to denote the two natural embeddings of (respectively) $\mch_{n-1,\eta-1}$ and $\mch_{n-1,\eta}$ into $\mch_{n,\eta}$ that come from the decomposition on the right hand side of Eq.~\eqref{eq:sd_split} (so that, for example, $J_1J_1^\dagger=N_u $ and $  J_0J_0^\dagger=\mathbf1-N_u$). 
The nice result is then that $f_{n,\eta}(J_1(x))=f_{n-1,\eta-1}(x)$, and $f_{n,\eta}(J_0(y))=f_{n-1,\eta}(y)$, for all $\ket x$ and $\ket y$ in their respective spaces (this follows from the calculations employed in the Appendix to prove Lemma~\ref{lem:slatnormal}, in particular Eq.~\eqref{eq:t2d}). That is, not only does the projector of our test block diagonalise with respect to the decomposition into $u$-occupancy sectors, but for a vector entirely in the two or zero $u$-occupancy sectors, we can equivalently factor off its $u$-dependence and calculate its acceptance probability using the test which works at ``the next level down'' (which, depending on     which occupancy sector we were in, will be either $\mch_{n-1,\eta-1}$ or $\mch_{n-1,\eta}$). This is good, but we are not quite done yet. By Eq.~\eqref{eq:sd_split} we can (for any fixed $u$) always write  a generic state in $\mch_{n,\eta}$ as $ \ket{\psi_{t_0}}=\sqrt {t_0}\,J_1(x)+\sqrt{1-t_0}\,J_0(y)$ for some $x$ and $y$ and then make our usual argument by promoting $t_0$ to a continuous parameter in $[0,1]$; the problem is that we still need to establish that there exists a $u$ for which one of the endpoints has at least as high a probability of acceptance as $\ket{\psi_{t_0}}$ itself does. This requires a little calculation, which we hide behind   the following lemma:
\begin{restatable}{lem}{lemslatenhance}\label{lem:slatenhance}
For every pure state $\ket\psi\in\mch_{n,\eta}$, there exists a  state $\ket{\phi_u}= a(u)\ket\psi /\norm{a(u)\ket\psi} $ represented by $u\in \mbc^n$ such that $f_{n-1,\eta-1}(\phi_u)\geq f_{n,\eta}(\psi)$.
\end{restatable}
Lemma~\ref{lem:slatenhance} is the last piece of information we need to be able to apply Theorem~\ref{thm:1}, which we can now do with $c=2/3$; this immediately establishes the present theorem. Finally, to establish optimality (at least in the nontrivial regime $\eta\neq0,1,n-1,n$) we can for example consider the state $\sqrt{1-\varepsilon^2}\ket{12}+\varepsilon\ket{34}$, which by standard arguments cannot be distinguished from $\ket{12}$ (which for $\varepsilon<1/2$ is its closest Slater state~\cite{zhang2016optimal}) with fewer than $\Om(\varepsilon^{-2}\log(1/\delta))$ copies. 
\end{proof}

\section{Discussion}\label{sec:discussion}
In this work, we have considered a general  route for studying property testing with respect to pure state families. To summarise, the high-level idea is that we write our state as a linear combination of the closest state within the family, and a state  departing orthogonally from it. By searching for  coordinates in which these two components have nice properties, we   hope   to find a ``witness vector'' which allows us to derive a bound on the rejection probability for states which are  not in the family, but have a non-negligible overlap with at least one state within it. Subsequently, we seek to use our path argument to upgrade this to a strong bound which applies non-vacuously to any state. \\

Our primary application is to establish that fermionic Gaussianity testing can be performed with a dimension-independent sample-complexity, which is readily seen to be asymptotically optimal. There are, however, a few natural questions remaining. First,  Theorem~\ref{thm:main} settles the copy complexity only up to constants, leaving open the possibility that more involved   collective tests might improve the precise number of copies required. Indeed, determining whether Bell-like tests achieve   optimal constants has been the subject of several recent works in quantum information theory~\cite{chang2026sample,beckey2026optimal}. 
Finally, although our second application, to   Slater determinant testing, has not received as much prior interest in the community, we find the fact that our general technique goes through very smoothly to a second example to be encouraging; certainly the extent to which it can be more broadly applied is an interesting open question.\\

\textit{Note added.} Constant sample-complexity gaussianity testing was also established (for bosons as well as fermions) in the concurrent works of Refs.~\cite{rad2026testing,iyer2026optimal}. The proofs techniques of these works are similar but not identical to ours; they consider the path from the unknown state to the target manifold along the flow of the gradient of the test rejection probability, and are able to prove the result by considering some of the properties of this path. We  encourage the interested reader to check out these nice works.

\section*{Acknowledgements}
MW and ML acknowledge support by the Laboratory Directed Research and Development (LDRD) program of Los Alamos National Laboratory (LANL) under project number 20260043DR, and by LANL’s ASC Beyond Moore’s
Law project. This work was also supported by the Quantum Science Center, a National Quantum Information Science Research Center of the U.S. Department of
Energy. 
The technical details leading to the proofs of our results were largely developed by GPT-6, with minimal human input. All AI-generated results were verified by the authors, and variously rewritten in order to meet our standards of clarity.

\begingroup
\raggedbottom
\bibliographystyle{IEEEtran}
\bibliography{quantum}
\clearpage
\endgroup
\appendix

\section{Further minutiae}\label{app:fm}
\noindent
In this appendix we give the proofs of the lemmas that were not proved in the main text.

\lemstationarity*
\begin{proof}
First of all,   the compactness of the Gaussian family  implies that a maximiser actually exists, and as the computational basis states span the space and are all Gaussian, its fidelity satisfies  $F>0$. As discussed in the main text, we can make a  ``Gaussian change of coordinates'' (i.e. multiplied by a suitable matchgate) so that this nearest Gaussian state is  the vacuum. This is helpful because we have a nice  chart for the (even) Gaussian family near the vacuum, namely
\begin{equation}\label{eq:thouless}
 \ket{g(z)}=\mathcal N(z)
 \exp\!\left(\sum_{i<j}z_{ij}a_i\ad a_j\ad\right)\ket0
 =\ket0+\sum_{i<j}z_{ij}\ket{ij}+\mco(\|z^2\|),
\end{equation}
where $\mathcal N(z)=1+\mco(\|z^2\|)$ and the $z_{ij}$ are
independent complex parameters. As in Section~\ref{sec:gen}, differentiating $\lvert\braket{g(z)}{\psi}\rvert^2$ with respect to both the real and imaginary parts
of every $z_{ij}$  forces the two particle amplitudes  $\ \braket{i,j|\psi}$ to vanish, which implies  the claim of the lemma.
\end{proof}

\lemnormal*
\begin{proof}
The cases $p=0,1$ are immediate, so let us assume $0<p<1$ and choose the vacuum amplitude to be $\sqrt p$. Our first observation is that $\Lm_n$  preserves the total occupation of each mode across the two copies, i.e., $[\Lm_n, (a_i\ad a_i)\ot \id+\id\ot(a_i\ad a_i)]=0$. Indeed, one can show  by a direct calculation that
\begin{equation*}
    \Lambda_n=2\sum_{i=1}^n \bigl(a_i\otimes a_i^\dagger+a_i^\dagger\otimes a_i\bigr),
\end{equation*}
and so can move particles between the two copies corresponding to a given mode, but cannot change the total number of particles at that mode.
For an even-sized set $K$ of size $k\geq4$, we  consider the (therefore $\Lm_n$-invariant) sector $\mch_K$ with total occupation one across the two copies   at each mode in $K$ and zero elsewhere; an orthonormal basis of this sector is
\begin{equation}\label{eq:sector}
 \{\ket S\otimes\ket{K\setminus S}:S\subseteq K\}
\end{equation}
(where $\ket S= \ket{s_1\dots s_n}$ is shorthand for the computational basis vector that satisfies $s_i=1$ iff $i\in S$). 
Now, let 
\begin{equation*}
  v_K= \sum_{S\subseteq K} \psi_S\psi_{K\setminus S} \,|S\rangle\otimes|K\setminus S\rangle  
\end{equation*}
be the projection of $\ket\psi^{\otimes2}$ into this sector. So, for example, $\psi_\varnothing=\sqrt p$, and the coefficient of every basis vector with $|S|=2$ or $k-2$ vanishes, because one of its factors is a forbidden two-particle amplitude. Now, through another direct calculation, one can see that
\begin{align}\label{eq:sectorbridge}
    \Lambda_K:= \Lambda_n \big\rvert_{\mch_K} &= 2 \sum_{j=1}^k (-1)^{j-1} X_{k_j} \otimes X_{k_j}\\
    \Lambda_K^2&= 4k\id_{\mch_K} + 8\sum_{i< j} (-1)^{i+j} (X_{k_i}X_{k_j}) \otimes (X_{k_i}X_{k_j}),
\end{align}
where $k=|K|$. We are ready to define our witness vector! Indeed, let $e_K=\ket0\otimes\ket K+\ket K\otimes\ket0,$ and $w_K=\Lambda_K^2e_K$ (note that since $\Lambda_K$ is Hermitian, $w_K$ is perpendicular to its kernel and hence belongs to the rejecting subspace of our Bell test). Now, because $w_K$ is a sum over vectors of the form $\ket S\ket T$ where $S$ and $T$ have support of magnitude in   $\{0,2,k-2,k\}$, and because (by the assumption of the lemma) $\psi_S=0$ for $|S|=2$, we have 
\begin{equation}\label{eq:witnessoverlap}
 \langle w_K,v_K\rangle=8k \sqrt p\,\psi_K.
\end{equation}
We next calculate the norm of our witness. 
Since $\Lambda_K$ is Hermitian, $\|w_K\|^2=\langle e_K,\Lambda_K^4e_K\rangle$. Writing $u=\ket0\ket K$ and $z=\ket K\ket0$, this expands into four terms:
\begin{equation*}
    \|w_K\|^2 = \langle u,\Lambda_K^4u\rangle+\langle z,\Lambda_K^4z\rangle+2\operatorname{Re}\langle u,\Lambda_K^4z\rangle.
\end{equation*}
The two diagonal terms are equal. To calculate either one of them, consider
\begin{equation*}
\Lambda_K^4 = 16 \sum_{i_1,i_2,i_3,i_4=1}^k \pm 
(X_{k_{i_1}}X_{k_{i_2}}X_{k_{i_3}}X_{k_{i_4}})\ts.
\end{equation*}
Now, a product of bit flips returns a computational-basis vector to itself exactly when every bit is flipped an even number of times. With four flips in total, there are only two possibilities: one index occurs four times ($k$ such cases), and two distinct indices each occur twice ($\binom{k}{2} \binom{4}{2}$ such cases).  It follows that $\langle u,\Lambda_K^4u\rangle = 16\left(k+6\binom{k}{2}\right) =16(3k^2-2k).$ For the off-diagonal term, we must transform $|0\rangle$ into $|K\rangle$, requiring that   every one of the $k$ bits has to be flipped (an odd number of times). If $k>4$, four flips cannot do this, so that $\langle u,\Lambda_K^4z\rangle=0$. When $k=4$, however, an off-diagonal contribution is possible, leading to $\langle u,\Lambda_K^4z\rangle=16 (4!)=384$ when $k=4$; in total we get  $\|w_K\|^2=2\cdot640+2\cdot384=2048$ when $k=4$. \\ 

\noindent
Now, since $Q_n$ is an orthogonal projector and $Q_nw_K=w_K$, we have
\begin{equation*}
    \lvert\langle w_K,v_K\rangle \rvert^2 = \lvert \langle w_K,Q_nv_K\rangle \rvert^2\leq  \|w_K\|^2\|Q_nv_K\|^2;
\end{equation*}
recalling Eq.~\eqref{eq:witnessoverlap} then yields
\begin{align*}
\|Q_nv_K\|^2 &\geq\frac{64k^2p|\psi_K|^2}{\|w_K\|^2}\\
&=\begin{cases}
\dfrac p2|\psi_K|^2,&k=4,\\[6pt]
\dfrac{2kp}{3k-2}|\psi_K|^2,&k>4.
\end{cases}\\
&\geq \frac p2|\psi_K|^2.
\end{align*}
This is a lower bound on the rejection probability of the projection of $\ket{\psi}\ts$ onto  $\mch_K$ sector. Now, if $K\neq K'$, then the sectors $\mathcal H_K$ and $\mathcal H_{K'}$ have different total-occupation eigenvalues for at least one mode, so they are orthogonal. Moreover, we showed earlier that $\Lm_n$, and hence $Q_n$, commutes with the projector onto each $\mch_K$, from which it follows that   the vectors $Q_nv_K$ are mutually orthogonal. With that in mind, and writing $\ket\psi\ts = \sum_Kv_K+v_{\mathrm{rest}}$, we have
\begin{equation*}
    r_n(\psi)=\bigl\|Q_n|\psi\rangle^{\otimes2}\bigr\|^2=\sum_K\|Q_nv_K\|^2+\|Q_nv_{\mathrm{rest}}\|^2\geq\frac p2\sum_K|\psi_K|^2=\frac{p(1-p)}2.
\end{equation*}
Finally, to apply all this to an arbitrary even pure state, we choose the  closest Gaussian (maybe it's not unique; choose one of them) and move it to the vacuum by a Gaussian unitary. The stationarity lemma then removes all two-particle amplitudes, and the vacuum probability becomes $p=F_{\mathcal G}(\psi)$. Gaussian invariance preserves the rejection probability, so the estimate becomes $r_n(\psi)\geq F_{\mathcal G}(\psi)(1-F_{\mathcal G}(\psi))/2$.

\end{proof}

\lembranching*
\begin{proof}
Let us fix an $a<b$ and an $s\in\{\pm1\}$. As in the main text, there is  a matchgate $V$ such that $Vc_aV^\dagger=c_1$ and $Vc_bV^\dagger=c_2$, so that $V\Pi_{ab,s}V^\dagger=\ketbra{q_s}\otimes \id$ for  $q_s=(1+s)/2$. As also discussed in the main text, a  pure state of definite total parity is therefore (up to the action of $V$) projected onto a tensor product   $\ket{q_s}\ket\phi$, where $\ket\phi$ is a pure $(n-1)$-mode state of definite parity.
Now, using the explicit forms of the Majoranas, we have (regardless of whether $\ket{q_s}$ is $\ket 0$ or $\ket 1$)
\begin{align*}
    \Lm_n (\ket{q_s}\ket\phi)\ts &= (X\ket{q_s})\ts \ket\phi\ts + (Y\ket{q_s})\ts \ket\phi\ts + (Z\ket{q_s})\ts \Lm_{n-1}\ket\phi\ts = \ket{q_s}\ts \Lm_{n-1}\ket\phi\ts,
\end{align*}
from which it follows that $f_n(\ket q_s\otimes\ket\phi)=f_{n-1}(\ket\phi)$, as the weight of $(\ket{q_s}\ket\phi)\ts$ killed by $\Lm_n$   is thus exactly the weight of $\ket\phi\ts$ killed by $\Lm_{n-1}$. Now,  as $P_n$ commutes with the matchgate group, it also commutes with the Lie algebra element $G_{ab}=ic_ac_b\otimes\id+\id\otimes ic_ac_b$. The eigenspace of $G_{ab}$ of
eigenvalue $2s$ has projector $\Pi_{ab,s}\otimes\Pi_{ab,s}$, which therefore also commutes with $P_n$. From Eq.~\eqref{eq:abs} it then follows that
\begin{align}
 \bra\psi^{\otimes2}P_n(\Pi_{ab,s}\otimes\Pi_{ab,s})\ket\psi^{\otimes2}&=\bra\psi^{\otimes2}(\Pi_{ab,s}\otimes\Pi_{ab,s})P_n(\Pi_{ab,s}\otimes\Pi_{ab,s})\ket\psi^{\otimes2}\nonumber \\ &=(\bra{q_s}\bra{\phi}\sqrt{a_{ab,s}})^{\otimes2}P_n(\sqrt{a_{ab,s}}|q_s\rangle\otimes|\phi_{ab,s}\rangle)^{\otimes2}\nonumber \\
 &=a_{ab,s}^{\,2}f_{n}(\ket{q_s}\ket{\phi_{ab,s}})\nonumber\\
 &=a_{ab,s}^{\,2}f_{n-1}(\ket{\phi_{ab,s}}).\label{eq:branchcompression}
\end{align}
Meanwhile, direct calculations reveal   that
\begin{equation}\label{eq:bridgesquare}
 \Lambda_n^2=2n\id+2\sum_{a<b}c_ac_b\otimes c_ac_b
\end{equation}
and
\begin{align*}
   \sum_{a<b,s}\Pi_{ab,s}\otimes\Pi_{ab,s} &=  \sum_{a<b}\big[\big(\frac{\id-ic_ac_b}{2}\big)\ts + \big(\frac{\id+ic_ac_b}{2}\big)\ts \big]\\
   &=  \frac 12\sum_{a<b}\big[\id + (ic_ac_b)\ts \big]\\
   &=  \frac 12 \binom{2n}{2} \id - \frac 12\sum_{a<b} (c_ac_b)\ts \\
   &=   \frac{n(2n-1)}{2}  \id - \big(\frac 14 \Lm_n^2 - \frac 12 n\id\big)  \\
   &=  n^2  \id - \frac 14 \Lm_n^2  ,
\end{align*}
where we have used Eq.~\eqref{eq:bridgesquare}. Introducing the normalised
\begin{equation}\label{eq:Sidentity}
 S_n=\frac1{n^2}\sum_{a<b,s}\Pi_{ab,s}\otimes\Pi_{ab,s}
 =\id-\frac{\Lambda_n^2}{4n^2},
\end{equation} 
we then have $0\leq S_n\leq\id$ and $P_nS_n=P_n$ (as by definition $\Lm_nP_n = 0$), whence
\begin{align*}
    f_n(\psi) &= \bra{\psi}\ts P_n \ket{\psi}\ts \\
     &= \bra{\psi}\ts P_nS_n \ket{\psi}\ts \\
     &= \bra{\psi}\ts P_n \Big(\frac1{n^2}\sum_{a<b,s}\Pi_{ab,s}\otimes\Pi_{ab,s}\Big) \ket{\psi}\ts \\
     &= \frac{1}{n^2}\sum_{a<b,s} a_{ab,s}^{\,2}f_{n-1}({\phi_{ab,s}})\\
     &\leq \Big(\max_{\substack{a<b,s\\a_{ab,s}>0}}f_{n-1}(\phi_{ab,s})\Big) \frac{1}{n^2} \sum_{a<b,s}a_{ab,s}^{\,2}  \\
     &\leq \max_{\substack{a<b,s\\a_{ab,s}>0}}f_{n-1}(\phi_{ab,s})  ,
\end{align*}
where we have used Eq.~\eqref{eq:branchcompression}  and that $\sum_{a<b,s}a_{ab,s}^{\,2} \leq n^2$, which follows from $S_n\leq \id$.  
\end{proof}

\lemslatnormal*
\begin{proof}
By the compactness of the  Slater family, we can find some closest   determinant state; by rotating (as in the fermionic Gaussian case) by $\Lm^\eta(v)$ for some $v$ we can without loss of generality take it to be  $\ket I = \ket{1,2,\ldots,\eta }$, with overlap $\inner I\psi=\sqrt F$. Now, by considering $\Lm^\eta(v)\ket{I}=(v\ket{1})\wedge \dots\wedge (v\ket{\eta})$, we see that the tangent space consists of single particle--hole substitutions:
\begin{equation}
 T_I\mcm_{n,\eta}
 =\Span_{\mathbb C}\{a_j\ad a_i\ket I:i\in I,\ j\notin I\},
\end{equation}
where we have introduced the shorthand notation $I=\{1,2,\ldots,\eta\}$. Indeed, expanding the action of $v$ on $\ket I$ in a Taylor series, we see that the first order terms are exactly given by replacing   one of the particles. As in the previous example, the optimality of $\ket I$ then forces such single-particle replacements to have zero amplitude in the part of $\ket\psi$ orthogonal to $\ket I$. That is, we can write $\ket\psi=\sqrt F\ket I+\sqrt{1-F}\ket\chi$ for $\ket\chi=\sum_{\substack{|K|=\eta,\ \ell_K\geq2}}\chi_K\ket K$, where $\ell_K=|I\setminus K|$. Now, recalling   $T_n=\sum_{i=1}^n a_i\ot a_i\ad$, and introducing   $e_K=\ket I\ket K+\ket{K}\ket{I}$, let us define the witness vector
\begin{equation}\label{eq:sd-witness}
\ket  w=T_{n,\eta}\ad T_{n,\eta}\sum_{\ell_K\geq2}\frac{\chi_K}{\ell_K}\ket{e_K}
\end{equation}
(note that this is very close in vibe-space to the witness vector of the fermionic Gaussian example). Now, because $T_n$ transfers only a single particle at a time from one copy to the other, it follows that $\inner w{\chi^{\otimes2}}=0$. Indeed, $\chi$ is a sum over computational basis states which differ in at least two places from $\ket I$, all of which are therefore orthogonal to $T_{n,\eta}\ad T_{n,\eta}\ket{I}\ket{K}$, and thence to $w$. This verifies one of the properties of Lemma~\ref{lem:witness} that we would like a witness vector to possess. Next, let us   calculate $\inner w{I\otimes\chi+\chi\otimes I}$. To begin, we can write
\begin{equation}
    T_{n,\eta}\ad T_{n,\eta} = \Big(\sum_{i=1}^n a_i\ad\ot a_i\big) \Big(\sum_{j=1}^n a_j\ot a_j\ad\big) = \sum_{i,j=1}^n (a_i\ad a_j)\ot (\delta_{ij} - a_j\ad a_i) = \eta\id\ts - \sum_{i,j=1}^n (a_i\ad a_j)\ot ( a_j\ad a_i).\label{eq:t2d}
\end{equation}
Writing $\ell=|I\setminus K|$ this yields, for example,
\begin{align}
    T_{n,\eta}\ad T_{n,\eta} \ket{I,K} &= \eta  \ket{I,K} - \Big(\sum_{i,j=1}^n (a_i\ad a_j)\ot ( a_j\ad a_i)\Big)\ket{I,K}\nonumber\\
     &= (\eta - |I\cap K|) \ket{I,K} - \Big(\sum_{i\neq j} (a_i\ad a_j)\ot ( a_j\ad a_i)\Big)\ket{I,K}\nonumber\\
     &= \ell \ket{I,K} - \sum_{i\neq j} \pm \ket{(I\cup\{i\})\setminus \{j\},(K\cup\{j\})\setminus \{i\}},\label{eq:tdecomp}
\end{align}
where the various $\pm$ signs are not something that we will need to compute. Now, there are exactly $\ell^2$ non-zero terms in the final sum, so that the magnitude squared of the above vector is $2\ell^2$, and
\begin{align*}
   \| T_{n,\eta}\ad T_{n,\eta} \ket{e_K}\|^2 &= \| T_{n,\eta}\ad T_{n,\eta} \ket{IK}\|^2  + \| T_{n,\eta}\ad T_{n,\eta} \ket{KI}\|^2 + 2 \mathfrak{Re}\braket{ T_{n,\eta}\ad T_{n,\eta}  IK }{ T_{n,\eta}\ad T_{n,\eta}  KI}\\
   &= 4\ell^2 + 2\, \mathfrak{Re}\braket{ T_{n,\eta}\ad T_{n,\eta}  IK }{ T_{n,\eta}\ad T_{n,\eta}  KI}.
\end{align*}
Now, it is easy to see that if $\ell>2$ the final overlap is zero, as one cannot move  enough particles to map $\ket{IK}$ and $\ket{KI}$ onto one another; if $\ell=2$, however, one has four ways of doing so. This yields
\begin{equation}
 \norm{T_{n,\eta}\ad T_{n,\eta} \ket{e_K}}^2 \leq
 \begin{cases}
 4\ell^2+8=6\ell^2,&\ell=2,\\
 4\ell^2,&\ell>2,
 \end{cases}
 \leq6\ell^2.
\end{equation}
This immediately yields
\begin{equation*}
 \norm {w}^2 \leq \sum_K\frac{|\chi_K|^2}{\ell_K^2}\norm{T_{n,\eta}\ad T_{n,\eta}\ket{e_K}}^2\leq6,
\end{equation*}
and
\begin{align*}
 \inner w{I\otimes\chi+\chi\otimes I}&=\sum_{\substack{K,K'\\ \ell_K,\ell_{K'}\geq 2}}\frac{\overline{\chi_K}\chi_{K'}}{\ell_K}\big(\bra{IK}+\bra{KI}\big)T_{n,\eta}\ad T_{n,\eta}\big(\ket{IK'}+\ket{K'I}\big)\\
 &=\sum_{\substack{K\\ \ell_K\geq 2}}\frac{|\chi_K|^2}{\ell_K}(2\ell_K)\\
 &=2,
\end{align*}
where we have used Eq.~\eqref{eq:tdecomp}.
This establishes the various  hypotheses of Lemma~\ref{lem:witness} with
$a=2$ and $B=6$, so that we conclude 
\begin{equation}
 r_{n,\eta}(\psi) \geq \frac {2F_{\mcs_{n,\eta}}(\psi)(1-F_{\mcs_{n,\eta}}(\psi))}{3}.
\end{equation}

\end{proof}

\lemslatenhance*
\begin{proof}
The plan is to deduce the existence of such a $u$ somewhat indirectly, by considering an inequality involving the average case, which will imply that at least some of the  $f_{n-1,\eta-1}(\phi_u)$  need to be greater than $ f_{n,\eta}(\psi)$.  To that end, we begin by recalling the elementary identity $\mathbb E_u[u_i\overline{u_j}u_p\overline{u_q}] =(\delta_{ij}\delta_{pq}+\delta_{iq}\delta_{pj})/[n(n+1)]$, where the average is over Haar random $u$~\cite{mele2023introduction}. This immediately implies  
\begin{equation*}
    \mathbb E_u[N_u\otimes N_u]=\frac{\eta^2\id+\Omega_{n,\eta}}{n(n+1)},
\end{equation*}
where $\Omega_{n,\eta}=\sum_{i,j} (a_i\ad a_j)\ot(a_j\ad a_i)$. Combined with Eq.~\eqref{eq:t2d}, this yields for $\alpha_{n,\eta}=n(n+1)/(\eta(\eta+1))$ that
\begin{equation*}
 S_{n,\eta}:=\alpha_{n,\eta}\mathbb E_u[N_u\otimes N_u]
 =\id-\frac{T_{n,\eta}\ad T_{n,\eta}}{\eta(\eta+1)},
\end{equation*}
where it can readily be further seen that $0\leq S_{n,\eta}\leq\id$ and $P_{n,\eta}S_{n,\eta}=P_{n,\eta}$. Now, let $p_u=\norm{a(u)\psi}^2=\inner\psi{N_u\psi}$, and set 
$\ket{\phi}_u=a(u)\ket \psi/\sqrt{p_u}\in\bigwedge^{\eta-1}u^\perp$, so that $N_u\ket \psi=\sqrt{p_u}\,J_1\ket{\phi_u}$. Since $D_u=N_u\otimes N_u$ projects onto the charge-two $u$-mode occupancy sector, it commutes with $P_{n,\eta}$ (as discussed in the main text). This then gives
\begin{align}
 f_{n,\eta}(\psi) &=\bra{\psi}\ts P_{n,\eta} \ket{\psi}\ts\nonumber\\
 &=\bra{\psi}\ts P_{n,\eta}S_{n,\eta} \ket{\psi}\ts\nonumber\\
& =\alpha_{n,\eta}\mathbb E_u\bra{ \psi}\ts D_uP_{n,\eta} D_u\ket{ \psi}\ts\nonumber \\
& =\alpha_{n,\eta}\mathbb E_u[ p_u^2 \bra{J_1 \phi_u}\ts P_{n,\eta} \ket{J_1 \phi_u}\ts ]\nonumber\\
 &=\alpha_{n,\eta}\mathbb E_u\left[p_u^2 f_{n-1,\eta-1}(\phi_u)\right],\label{eq:ave}
\end{align}
where we have used that the projector $D_u$ satisfies $D_u^2=D_u$, and that $f_{n,\eta}(J_1(x))=f_{n-1,\eta-1}(x)$ (see the main text). But we also know that $\alpha_{n,\eta}\mathbb E_u p_u^2=\bra{\psi}\ts S_{n,\eta} \ket{\psi}\ts\leq1$ (as $S\leq \id$) from which it follows that there must be some  $u$ with
$p_u>0$ and $f_{n-1,\eta-1}(\phi_u)\geq f_{n,\eta}(\psi)$, lest we contradict Eq.~\eqref{eq:ave}.

\end{proof}

\end{document}